\documentclass[%
reprint,
 amsmath,
 amssymb,
 aps,
 unsortedaddress,
 pra,
]{revtex4-2}

\usepackage{graphicx}% Include figure files
\usepackage{dcolumn}% Align table columns on decimal point
\usepackage{bm}% bold math
\usepackage{amsmath,amssymb,amsthm,easybmat,verbatim}
\usepackage{color}
\usepackage{xcolor}
\usepackage{tikz}
\usepackage{algpseudocode}
\usepackage{algcompatible}

\newlength\myindent
\newtheorem*{prop*}{Proposition*}
\newtheorem{prop}{Proposition}

\usepackage[ruled, linesnumbered]{algorithm2e}

\usepackage{orcidlink}
\usepackage{bm}

\usepackage{physics}
\usepackage[caption=false]{subfig}
\usepackage{hyperref}

\hypersetup{
    colorlinks = true,
    citecolor = blue,
    urlcolor = blue
}

\usepackage{accents}

\newcommand{\bfly}[2]{\ket{#1}\!\!\bra{#2}}

\begin{document}

\preprint{APS/123-QED}

\title{Efficient Sparse State Preparation via Quantum Walks}

\author{Alvin Gonzales$^{1}$\orcidlink{0000-0003-1635-106X}}
% \thanks{Contributed equally}
\email{agonza@siu.edu}
\author{Rebekah Herrman$^{2}$\orcidlink{0000-0001-6944-4206}}
% \thanks{Contributed equally}
% \email{rherrma2@utk.edu}
\author{Colin Campbell$^3$\orcidlink{0009-0009-4261-6878}}
% \email{colin.campbell@infleqtion.com}
\author{Igor Gaidai$^2$\orcidlink{0000-0002-3950-3356}}
% \email{gaidai92@gmail.com}
\author{Ji Liu$^1$\orcidlink{0000-0002-5509-5065}}
% \email{ji.liu@anl.gov}
\author{Teague Tomesh$^3$\orcidlink{0000-0003-2610-8661}}
% \email{teague.tomesh@infleqtion.com}
\author{Zain H. Saleem$^1$\orcidlink{0000-0002-8182-2764}}
% \email{zsaleem@anl.gov}
\affiliation{$^1$Mathematics and Computer Science Division, Argonne National Laboratory, Lemont, IL, 60439, USA}
\affiliation{$^2$Department of Industrial and Systems Engineering, University of Tennessee Knoxville, Knoxville, TN, 37996, USA}
\affiliation{$^3$Infleqtion, Chicago, IL, 60604, USA}

\date{\today}% It is always \today, today,
             %  but any date may be explicitly specified

\begin{abstract}
\noindent
%Continuous-time quantum walks (CTQWs) on dynamic graphs, referred to as dynamic CTQWs, are a recently introduced universal model of computation that offers a new paradigm in which to envision quantum algorithms. In this work we develop an algorithm that converts single-edge and self-loop dynamic CTQWs to the gate model of computation. We use this mapping to introduce an efficient sparse quantum state preparation framework based on dynamic CTQWs. 
%Our approach utilizes global information about the target state, relates state preparation to finding the optimal path in a graph, and leads to optimizations in the reduction of controls that are not as obvious in other approaches.  
%Interestingly, classical optimization problems such as the minimal hitting set, minimum spanning tree, and shortest Hamiltonian path problems arise in our framework. 
%We test our methods against the state-of-the-art ancilla free sparse state preparation method and a uniformly controlled rotation method, which is used by Qiskit, and find ours requires fewer CX gates when the target state has a polynomial number of non-zero amplitudes.

Continuous-time quantum walks (CTQWs) on dynamic graphs, referred to as dynamic CTQWs, are a recently introduced universal model of computation that offers a new paradigm in which to envision quantum algorithms. In this work we develop an algorithm that converts single-edge and self-loop dynamic CTQWs to the gate model of computation. We use this mapping to introduce an efficient sparse quantum state preparation framework based on dynamic CTQWs. 
Our approach utilizes combinatorics techniques such as minimal hitting sets, minimum spanning trees, and shortest Hamiltonian paths to reduce the number of controlled gates required to prepare sparse states. We show that our framework encompasses the current state of the art ancilla free sparse state preparation method by reformulating this method as a CTQW. This CTQW-based framework offers an alternative to the uniformly controlled rotation method used by Qiskit by requiring fewer CX gates when the target state has a polynomial number of non-zero amplitudes.

\end{abstract}

%\keywords{Suggested keywords}%Use showkeys class option if keyword
                              %display desired
\maketitle
\makeatother
\newtheorem{definition}{Definition}[section]
\newtheorem{assumption}{Assumption}[section]
\newtheorem{theorem}{Theorem}[section]
\newtheorem{lemma}{Lemma}[section]
\newtheorem{conjecture}{Conjecture}[section]
\newtheorem{property}{Property}[section]

\section{Introduction}
\label{sec:intro}

Quantum algorithms have the potential to provide speed-up over classical algorithms for some problems \cite{montanaro2015quantum, shor1999polynomial}. However, certain quantum algorithms may require non-trivial input states \cite{harrow2009quantum, larose2020robust}, which in general, are challenging to prepare and may require exponentially many CX gates \cite{Mottonen_2005TransformOfQuantStsUsingUnifContRot} in the worst case. 
%Nevertheless, many practically relevant quantum states can be prepared much more efficiently by taking advantage of their special properties or their sparsity \cite{Gleinig_2021AnEffAlgoForSparseQSP}. 
%Due to its fundamental role in quantum algorithms \cite{vazquez2021efficient, zhang2022_qspWithOptimalCircDepthImplemAndApps, gard2020efficient, harrow2009quantum, larose2020robust}, quantum state preparation is an active area of research within the space of quantum computation.
Quantum state preparation (QSP) is an active area of research within the space of quantum computation \cite{vazquez2021efficient, zhang2022_qspWithOptimalCircDepthImplemAndApps, gard2020efficient}.
The resource overhead of general QSP problems for arbitrary states is known to be $\mathcal{O}(2^n)$ in the number of CX gates, where $n$ is the number of qubits \cite{Mottonen_2005TransformOfQuantStsUsingUnifContRot, Shende_2006SynthOfQuantLogCirc}. Within this limit, there are different techniques utilizing different resources, which can be useful depending on the context. For instance, if depth is more of a concern than space, Ref. \cite{Araujo_2021ADivAndConqAlgoForQSP} provides an $\mathcal{O}(n^2)$ depth divide-and-conquer method to prepare arbitrary states at the cost of needing $\mathcal{O}(2^n)$ ancillary qubits. In Ref.~\cite{Gui_2024_Spacetime-Eff}, the ancilla overhead is improved to $\mathcal{O}(n)$. For arbitrary QSP, the state of the art deterministic protocol achieves $\mathcal{O}(\frac{23}{24}2^n)$ CX scaling on even numbers of qubits \cite{Plesch_2011QSPWithUniversalGateDecomp}. The downside of this method is that it relies on the Schmidt decomposition of the $n$-qubit target state, a computationally expensive task, which has a classical run time of $\mathcal{O}(2^{3n/2})$ \cite{golub1996matrix}.

While arbitrary state preparation is exponential, there are states of practical interest that do not require exponential resources, even in the worst case \cite{cruz2019efficient, Bartschi_2019DetermPrepOfDickeSts}. Manual methods are strategies for QSP that use advanced knowledge of the target state to save resources in ways that might not be obvious when preparing arbitrary states. For example, GHZ states are highly entangled, but only require a linear number of CX gates. Similarly, the methods introduced in Ref.~\cite{bartschi2022short} require $\mathcal{O}(kn)$ CX gates to prepare $n$-qubit Dicke states of Hamming weight $k$. For machine learning applications on classical data, quantum data loaders can prepare sparse amplitude encodings of $d$-dimensional real-valued vectors on $d$-qubits in $\mathcal{O}(\log d)$ circuit depth \cite{kerenidis2024quantum}. In Ref.  \cite{Zhang_2024QuantCompMethForSolvEMProbBasedOnTheFiniteElemMeth}, Zhang et al. introduce a technique to prepare input states to the HHL algorithm \cite{harrow2009quantum} based on the finite element method equations relevant to electromagnetic problems. Zhang et al.'s algorithm achieves $\mathcal{O}(n)$ depth scaling by exploiting the symmetry and sparsity of the relevant target states; however, these useful assumptions are not guaranteed for arbitrary states.

The state of the art ancilla free method for preparing arbitrary sparse states consisting of $m\ll 2^n$ non-zero amplitude computational basis states is detailed in Ref.~\cite{Gleinig_2021AnEffAlgoForSparseQSP}, which produces circuits with $\mathcal{O}(nm)$ CX gates and runs in $\mathcal{O}(nm^2\log (m))$ time classically. This type of method is attractive for real-world applications of quantum computers because they can take advantage of sparsity while still remaining agnostic to particular details of the state. Other methods achieve this, but require ancillas \cite{de_Veras_2022_double_sparse_qsp, Mozafari_2022EffDetPrepQuantStUsingDecisionDiag}. The method using decision trees \cite{Mozafari_2022EffDetPrepQuantStUsingDecisionDiag} requires 1 ancilla qubit. The method in Ref.~\cite{Malvetti_2021QuantCircsForSparseIsom} also achieves $\mathcal{O}(nm)$ CX gates. 
% and runs in, for sparse states, $\mathcal{O}(\binom{n}{\log(m)}+nm^2$) classically. 
However, this method is based on Householder reflections and is more complicated than the one introduced in Ref.~\cite{Gleinig_2021AnEffAlgoForSparseQSP}.

In this work we show how continuous-time quantum walks (CTQWs) on dynamic graphs can be used to create arbitrary sparse (and dense) quantum states. CTQWs on graphs are a universal model of computation \cite{childs2009universal} that excels at spatial searches \cite{farhi1998quantum, childs2004spatial} and has applications in finance \cite{scalas2006application}, coherent transport on networks \cite{mulken2011continuous}, modeling transport in geological formations \cite{berkowitz2006modeling}, and combinatorial optimization \cite{qiang2012enhanced}.
In this model, the quantum state vector is evolved under the action of $e^{-i A t}$ for some time $t$, where $A$ is an adjacency matrix of some fixed unweighted graph.
Such walks have been implemented natively on photonic chips \cite{chapman2016experimental, qiang2016efficient, tang2018experimental}.

In 2019, CTQWs on \textit{dynamic} graphs (i.e. graphs that may change as a function of time) were introduced and shown to also be universal for computation \cite{herrman2019continuous} by implementing the universal gate set H, CX, and T.

In the original dynamic graph model, isolated vertices were propagated as singletons. However a new model where isolated vertices are not propagated was introduced later in Ref.~\cite{wong2019isolated}. Simplification techniques were introduced that can reduce the length of the dynamic graph sequence if graphs in the sequence satisfy certain properties \cite{herrman2022simplifying}.
Furthermore, the authors of \cite{adisa2021implementing} found that CTQWs on at most three dynamic graphs can be used to implement the equivalent of a universal gate set.

While prior work has focused on representing quantum gate operations as CTQWs, the opposite conversion -- generating a quantum circuit from a CTQW -- is less well-studied. Since CTQWs on dynamic graphs have yet to be natively implemented on hardware, and given recent progress in the development of gate-based quantum hardware \cite{radnaev2024universal, bluvstein2024logical, wang2024fault, acharya2024quantum}, converting CTQWs to the gate model is necessary to execute them on existing quantum hardware. Additionally, developing an algorithm for such a conversion offers an alternative way of thinking about circuit design, which may result in simpler circuits, compiler optimization techniques, and new ansatzes for QAOA \cite{farhi2014quantum, zhu2022adaptive, liu2022quantum, wilkie2024qaoa}, MA-QAOA \cite{herrman2022multi, gaidai2023performance, shi2022multiangle} or VQE \cite{peruzzo2014variational, tang2021qubit, liu2022layer}, where the trainable parameters correspond to dynamic graph propagation times.

To this end, we develop an algorithm that converts CTQWs on single-edge graphs and single self-loop graphs, which can serve as a basis for arbitrary graphs, to a sequence of gates in the circuit model. 
This conversion algorithm is used as a foundation for a new deterministic arbitrary quantum state preparation framework, where the CX count of the resulting circuit is linear in the number of $m$ non-zero amplitude computational basis states that comprise the target state and the number of qubits $n$. Fig.~\ref{fig:walks-to-circuit} gives a general picture of the quantum walks QSP framework introduced in this paper. Our framework does not require ancillas.

\begin{figure}
    \centering
    \includegraphics[width=\linewidth]{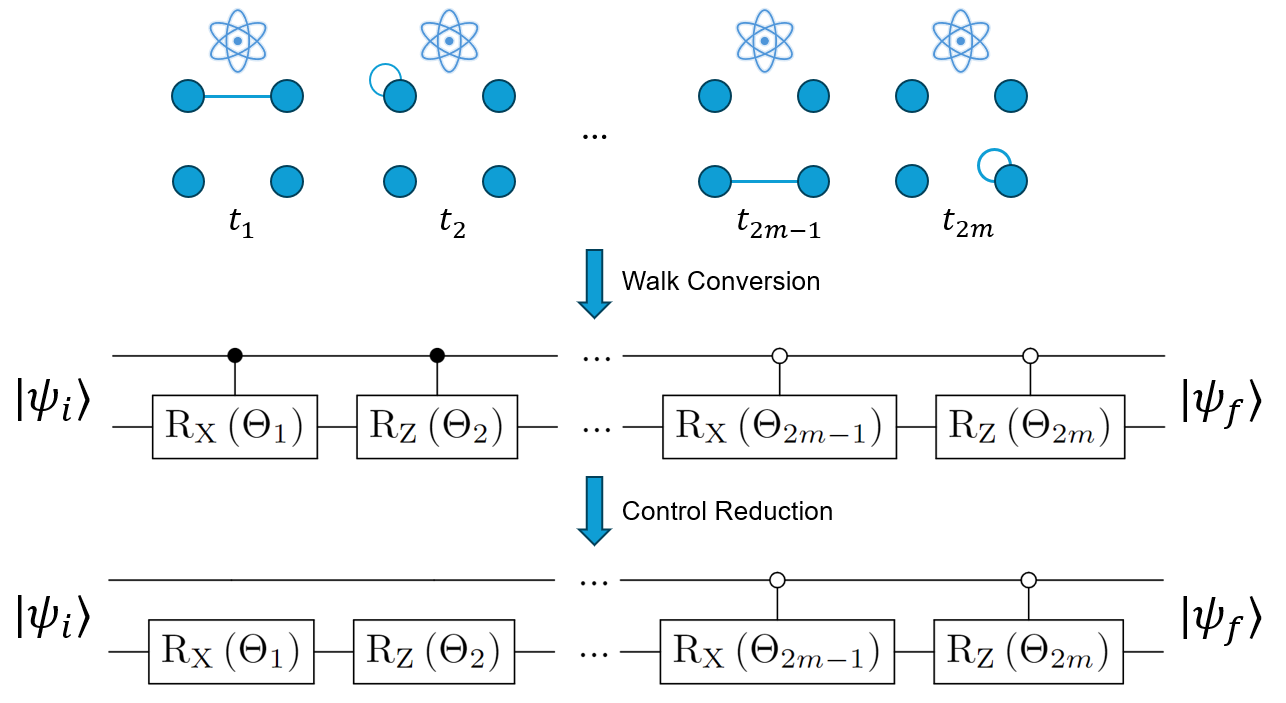}
    \caption{Alternating single-edge and self-loop quantum walks yields a framework for arbitrary QSP. We provide conversion methods to construct the gate-based circuit.}
    \label{fig:walks-to-circuit}
\end{figure}

Our approach builds on the idea of exploiting sparsity to increase the efficiency of QSP techniques by approaching the problem from the perspective of dynamic quantum walks. 
As such, our method is best suited for the asymptotically sparse states ($m = \mathcal{O}(poly(n))$), but can, in principle, also work for the asymptotically dense states ($m = \mathcal{O}(2^n)$).
The walk framework that we present here is flexible, scalable, and intuitive, building upon well-established graph-based algorithms to efficiently prepare sparse quantum states. We also show that the state of the art ancilla free merging states method of Ref.~\cite{Gleinig_2021AnEffAlgoForSparseQSP} is encompassed in our CTQW state preparation framework.

% \subsection{Continuous-Time Quantum Walks on Dynamic Graphs}

A \textit{dynamic graph} is defined as a sequence of ordered pairs $\{ (G_i, t_i) \}_{i =1}^\ell$ for some $\ell \in \mathbb{N}$ that consists of unweighted graphs $G_i$ and corresponding propagation times $t_i$. A CTQW on a dynamic graph is then a CTQW where the first walk is performed on graph $G_1$ for time $t_1$, followed by a walk on graph $G_2$ for time $t_2$, and so on until all graphs in the sequence have been walked upon. The final state $\ket{\psi_\ell}$ is then given by
\begin{equation}
\label{eq:walk_state}
    \ket{\psi_\ell}=e^{-i A_\ell t_\ell} \ldots e^{-iA_2t_2}e^{-i A_1 t_1} \ket{\psi_0}.
\end{equation}
where $A_i$ is the adjacency matrix for graph $G_i$, that is a 0-1 valued symmetric matrix.
For simplicity, we assume that each graph has $2^n$ vertices for some $n \in \mathbb{N}$, and each vertex represents a computational basis state.

As an example, consider the dynamic graph found in Fig.~\ref{fig:cx}, where the adjacency matrices are 
\begin{align*}
    A_1 &= 
    \begin{pmatrix}
        0& 0& 0& 0\\
        0& 0& 0& 0\\
        0& 0& 0& 1\\
        0& 0& 1& 0 
    \end{pmatrix}\\
    A_2 &=
    \begin{pmatrix} 
        0& 0& 0& 0\\
        0& 0& 0& 0\\
        0& 0& 1& 0\\
        0& 0& 0& 1 
    \end{pmatrix}
\end{align*}

\begin{figure}
    \includegraphics{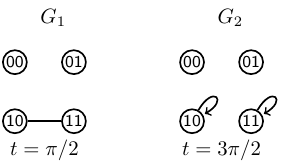}
    \caption{Dynamic CTQW implementation of CX gate.}
    \label{fig:cx}
\end{figure}

If the walker starts in the initial state $\ket{\psi_0} = (a,b,c,d)^T$, then the final state of this walker is 

\begin{equation*}
    \ket{\psi_2} = 
    \begin{pmatrix} 
        1& 0& 0& 0\\
        0& 1& 0& 0\\
        0& 0& i& 0\\
        0& 0& 0& i 
    \end{pmatrix}
    \begin{pmatrix} 
        1& 0&  0&  0\\
        0& 1&  0&  0\\
        0& 0&  0& -i\\
        0& 0& -i&  0 
    \end{pmatrix}
    \begin{pmatrix} 
        a\\
        b\\
        c\\
        d
    \end{pmatrix}
    =
    \begin{pmatrix} 
        a\\
        b\\
        d\\
        c
    \end{pmatrix}
\end{equation*}
which is equivalent to a CX gate in the circuit model.

Interestingly, CTQWs on dynamic graphs that correspond to operations in the gate model tend to have a similar form: some phase is added to particular basis states via self-loops on vertices of the graph, then they are connected by edges to allow for state transfer, followed by additional self-loops to eliminate unwanted phase.
This general alternating sequence of graphs inspires the state preparation algorithm introduced later in this work.

This paper is organized as follows. In Sec.~\ref{sec:results}, we present an algorithm for converting single-edge and single self-loop CTQWs to the circuit model and give a general overview of our deterministic state preparation method. Next, we compare various walk orders to the sparse state preparation method from Ref.~\cite{Gleinig_2021AnEffAlgoForSparseQSP} and to the uniformly controlled rotation method \cite{Mottonen_2005TransformOfQuantStsUsingUnifContRot, Shende_2006SynthOfQuantLogCirc} used by Qiskit. In both cases, we found walk orders that require fewer CX gates. Then we discuss the impact of the results and directions for future work in Sec.~\ref{sec:discussion}. Finally, combinatorics methods for selecting walking orders that reduce the number of controlled gates in the QSP circuits and additional optimizations are presented in Sec.~\ref{sec:methods}.

\section{Results}
\label{sec:results}

% \subsection{Conversion from CTQW on dynamic graphs to the Circuit Model}
% \label{sec:dict}

A standard basis for any adjacency matrix is given by the set of adjacency matrices corresponding to all single-edge walks and all single self-loop walks. We introduce the conversion methods that can construct the gate-based representations for these walks on $n$ qubits. Throughout this paper, the considered states are represented in the standard computational basis set.

% \subsubsection{Single-Edge CTQW to Circuit Model}
% \label{sec:single_edge_conversion}
First, we introduce the conversion from single-edge CTQWs to the circuit model.
The adjacency matrix for a single-edge graph connecting basis states $\ket{j}$ and $\ket{k}$ is given by
\begin{align}
    A(j,k) = \bfly{j}{k} + \bfly{k}{j}
\end{align}
and the corresponding CTQW for time $t$ is
\begin{multline}\label{eq:singleEdgeU}
    U(j,k;t) = e^{-itA(j,k)} = \cos(t)(\op{j} + \op{k}) \\- i\sin(t)(\bfly{j}{k} + \bfly{k}{j}) + \sum_{l \notin \{j, k\}} \op{l}
\end{multline}
This unitary transfers amplitude between states $\ket{j}$ and $\ket{k}$ and leaves all other states untouched. 

If $\ket{j}$ and $\ket{k}$ differ in exactly one bit at position $l$, $U(j, k; t)$ is the $(n-1)$-controlled Rx($2t$) gate
\begin{align}
  \text{Rx}(2t) = 
  \begin{pmatrix} 
    \cos(t)   & -i\sin(t) \\
    -i\sin(t) & \cos(t) 
  \end{pmatrix}
\end{align}
where the target qubit is $l$, and the remaining qubits are either 0- or 1-controls corresponding to the remaining bits of $\ket{j}$ and $\ket{k}$.

If the Hamming distance between $\ket{j}$ and $\ket{k}$ is greater than $1$, using the CRx($2t$) gate is still possible, but requires some extra overhead. First apply a sequence of CX gates where the control is any bit in which $\ket{j}$ and $\ket{k}$ differ and the targets are the other bits in which $\ket{j}$ and $\ket{k}$ differ. Then apply the CRx($2t$) gate as before and apply the reverse of the conjugating CX sequence. An example conversion of a single-edge walk to gates is shown in Fig.~\ref{fig:adj-to-circuit}. The first CX gate transforms the two basis states such that their Hamming distance becomes equal to 1. The CRx creates the superposition and the final CX restores the original basis states. 

The $(n - 1)$-controlled CRx gate is a multi-controlled special unitary, and can be decomposed in $\mathcal{O}(n)$ CX and single-qubit gates \cite{vale_2023decompOfMultiContSpecUnitarySingleQubGates}.

\begin{figure}
    \centering
    \includegraphics[width=0.42\textwidth]{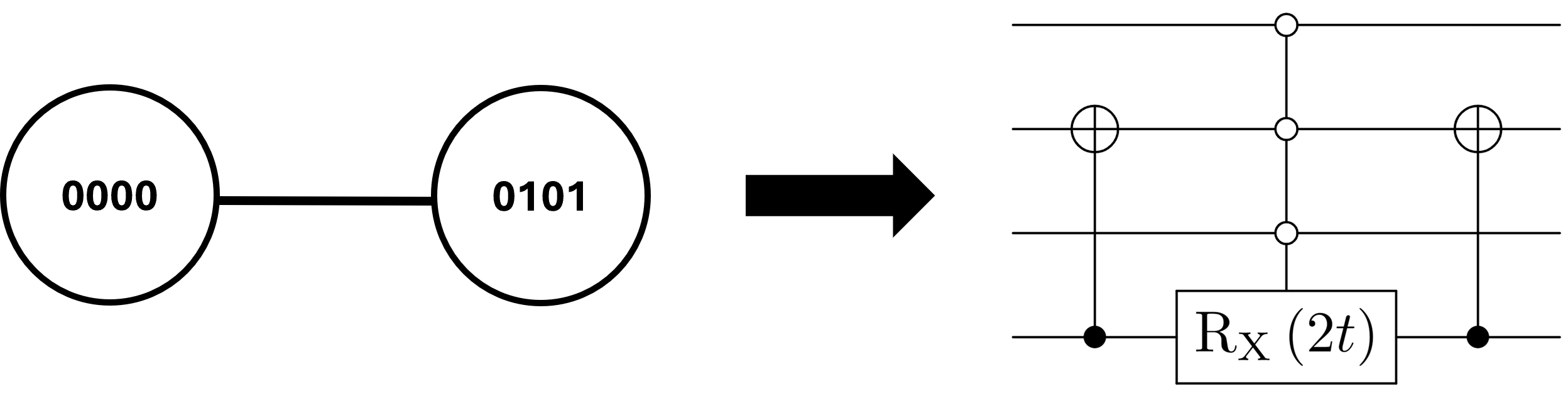}
    \caption{When a single-edge walk is between states with non-unit Hamming distance from each other, CX gates might be required in addition to the CRx gate.}
    \label{fig:adj-to-circuit}
\end{figure}

% \subsubsection{Single Self-Loop CTQW to Circuit Model}
% \label{sec:SingleSelfLoppCTQWtoCirc}
Next, we introduce the conversion from single self-loop CTQWs to the circuit model.
A single self-loop walk on basis state $\ket{j}$ is given by the adjacency matrix

\begin{align}
    A(j) = \op{j},
\end{align}
The corresponding CTQW for time $t$ is
\begin{align}\label{eq:selfLoop1}
    U(j;t) = e^{-itA(j)} = e^{-it}\op{j} + \sum_{k \neq j} \op{k},
\end{align} 
which is a diagonal matrix with a single non-unit element on the diagonal.

From this equation, one can see that the self-loop walk on a single vertex is equivalent to adding phase to exactly one computational basis state. This operation corresponds to an application of an $(n - 1)$-controlled phase gate P for time $-t$
\begin{align}
  \text{P}(t) = 
  \begin{pmatrix}
    1 & 0 \\
    0 & e^{it}
  \end{pmatrix},
\end{align}
where the target qubit can be arbitrarily chosen among the bits of $j$ that are equal to 1 (or 0, with the appropriate X conjugation), and the remaining qubits are either 0- or 1-controls corresponding to the remaining bits of $\ket{j}$.

However, the $\text{P}(t)$ gate is not a special unitary. As such, without using ancilla qubits, its decomposition requires $\mathcal{O}(n^2)$ CX gates \cite{da2022linear}. A more efficient (but also more nuanced) approach is to use an $(n - 1)$-controlled Rz($2t$) gate
\begin{align}
  \text{Rz}(2t) = 
  \begin{pmatrix}
    e^{-it} & 0 \\
    0           & e^{it}
  \end{pmatrix}
\end{align}
which is a special unitary and can be decomposed in $\mathcal{O}(n)$ CX gates \cite{vale_2023decompOfMultiContSpecUnitarySingleQubGates}.

At a first glance, an $(n - 1)$-controlled Rz gate affects two adjacent computational basis states, which makes it not equivalent to a self-loop walk in general. However, if only one of these two basis states exists in the state affected by the CRz gate, then it becomes essentially equivalent to the CP gate and can also be used to implement a self-loop walk. 

As long as the state we apply the walk to ($\ket{\psi_0}$ in Eq.~\ref{eq:walk_state}) has at least one zero-amplitude basis state $\ket{z}$, not necessarily adjacent to $\ket{j}$, one can use the same CX conjugation technique as described in the previous section for the Rx gate. This enables the interaction between $\ket{j}$ and $\ket{z}$ via the CRz gate and implements the self-loop walk on $\ket{j}$.

In the case when $\ket{\psi_0}$ is fully dense, i.e. all $2^n$ basis states have non-zero amplitudes, the above method will not work and a CP gate will have to be used instead.

% \subsection{Universal Computation}
As it was mentioned in the Introduction, CTQWs on dynamic graphs are known to be universal \cite{herrman2019continuous}. However, the authors of Ref. \cite{herrman2019continuous} use walks on arbitrary graphs in their proof of universality. In this section we prove that using only single-edge and single self-loop walks is sufficient to decompose an arbitrary unitary.
\begin{prop}
    An arbitrary $d \times d$ unitary can be decomposed into a series of single self-loop and single-edge CTQWs.
\end{prop}
\begin{proof}
    The case of $d=1$ is trivial so we consider $d\geq 2$. An arbitrary unitary can be decomposed into a series of 2-level unitaries (unitaries that act non-trivially on two or fewer basis states, pages 189-191 of \cite{nielsen2011quantumCompAndQuantInfo}).
    
    Similarly to how an arbitrary single-qubit unitary can be decomposed as $U = W(\alpha) \text{Rz}(\beta) \text{Rx}(\gamma) \text{Rz}(\delta)$ (page 175 of \cite{nielsen2011quantumCompAndQuantInfo}), where $W=e^{i\alpha}I$, an arbitrary 2-level unitary $U$ can be decomposed as
    \begin{equation}
    \label{eq:2by2Decomp}
        U = \text{W'}(\alpha)\text{Rz'}(\beta) \text{Rx'}(\gamma) \text{Rz'}(\delta)
    \end{equation}
    for some real numbers $\alpha, \beta, \gamma, \delta$, where the W', Rz' and Rx' are 2-level unitaries 
    whose action in the corresponding 2D subspace is equivalent to their single-qubit counterparts.

    As it was shown in the introduced conversions from CTQWs to the circuit model, Eq.~\eqref{eq:singleEdgeU} and  Eq.~\eqref{eq:selfLoop1}, a single-edge CTQW corresponds to Rx', and a sequence of two single self-loop CTQWs corresponds to W' and Rz'.
    Thus, an arbitrary unitary can be decomposed into a series of single self-loop CTQWs and single-edge CTQWs.

\end{proof}

% \subsection{Deterministic state preparation via quantum walks}
% \label{sec:general_approach}
In this section, we describe how a series of CTQWs on dynamic graphs can be used to prepare an arbitrary quantum state. 

% Next, we compare different state preparation ordering heuristics (details and optimizations found in Sec.~\ref{sec:methods}) that we developed to determine which requires the fewest controlled gates. We then compare these methods to the state-of-the-art state preparation method found in \cite{Gleinig_2021AnEffAlgoForSparseQSP} and to Qiskit built-in state preparation methods.

% \subsubsection{General Approach to State Preparation}
Given the task of preparing a quantum state with $m < 2^n$ non-zero amplitudes on $n$ qubits, we start by presenting a high-level overview of the framework to accomplish this task:
\begin{enumerate}
    \item Construct a tree graph connecting the non-zero amplitudes of the target state in some way. 
    % \textcolor{red}{See Fig.~\ref{fig:sparse-subset} for an example. UPDATE THIS AFTER NEW FIGURE 4 IS INSERTED}
    \item Starting from some node, and following a graph traversal order (see Sec.~\ref{sec:methods}), perform a self-loop walk for each encountered node and a single-edge quantum walk for each encountered edge.
    \item Convert the sequence of the quantum walks to the circuit representation (which may include control reduction and CX forward- and backward-propagation, as described below).
\end{enumerate}

Many different heuristics can be used to construct the sequence of walks in step 1 or choose a particular traversal order in step 2. Initially, the system starts with all amplitude in the root state, which can be constructed from the ground $\ket{0}^{\otimes n}$ state by applying the X gates to the necessary qubits.

An arbitrary amount of amplitude (absolute value) from the source can be transferred to the destination for any pair of the connected states corresponding to each single-edge walk, and all non-zero amplitude states are connected transitively via a tree. Therefore, an arbitrary distribution of absolute values of amplitude (but not phases) can be established on the set of non-zero amplitude basis states via single-edge walks.

Each self-loop walk on a basis state $\ket{z_i}$ establishes the correct phase of the corresponding coefficient $c_i$, but does not change the absolute value of it. Therefore, performing self-loop walks on each non-zero amplitude basis state establishes the correct phases of the corresponding coefficients, without interfering with the absolute values established by the single-edge walk. Hence, an arbitrary quantum state can be prepared with a sequence of single-edge walks and self-loop walks, as described in the algorithm above.

In the context of the state preparation, CRz gates can easily be used instead of CP gates to implement the self-loop walks, since every single-edge walk transfers the amplitude to a new basis state with zero amplitude. Thus, as described in the previous section, the same basis state can be used to establish the correct phase on the source state for the single-edge walk without additional CX conjugations. However, this will not apply to the leaf nodes (degree one vertices) in the tree, which will require another zero-amplitude state to interact with.

The tree contains $\mathcal{O}(m)$ nodes and edges, therefore $\mathcal{O}(m)$ walks will be required. As described in the previous section, each walk is represented by a single multi-controlled gate (CRz or CRx) that has up to $n - 1$ controls. Each such gate can be implemented in $\mathcal{O}(n)$ CX gates \cite{vale_2023decompOfMultiContSpecUnitarySingleQubGates} and may need to be conjugated by $\mathcal{O}(n)$ additional CX gates due to the Hamming distance between the adjacent basis states. Therefore the overall complexity of the algorithm, in terms of the required number of CX gates, is $\mathcal{O}(nm)$.

\begin{figure}
    \centering
    \includegraphics[width=0.45\textwidth]{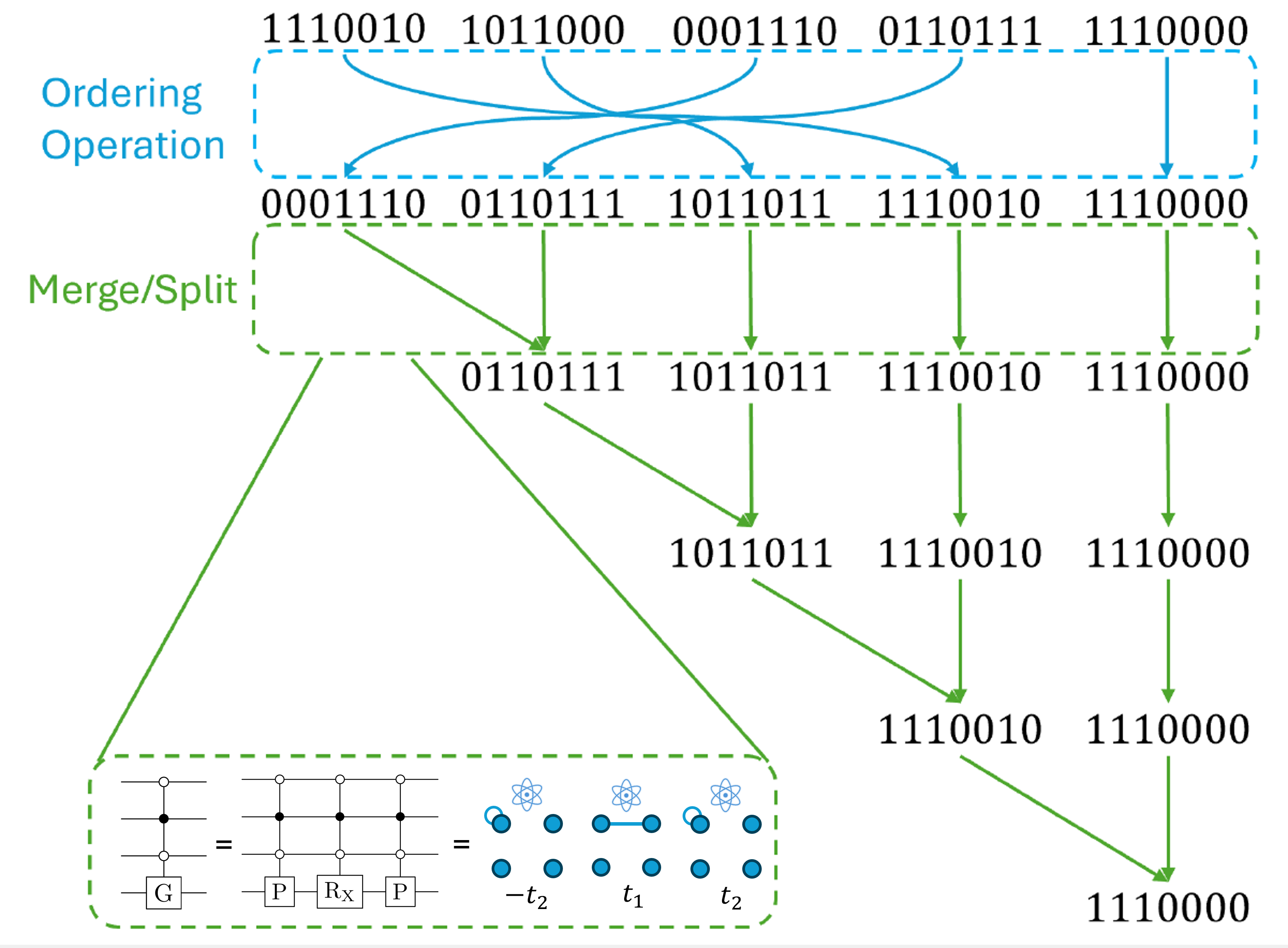}
    \caption{Given a set of bit strings over which to prepare the superposition state, this walk-based family of methods applies a heuristic for choosing an ordering of the basis states and then applies a sequence of merging (or equivalently splitting) operations to reduce (expand) an initial state to a desired endpoint. Each merge (split) operation can be interpreted as a sequence of quantum walks. Ref.~\cite{Gleinig_2021AnEffAlgoForSparseQSP} uses the controlled G Gate as the primary merging operation, but this gate is equivalent to the action of a controlled P, controlled Rx, and controlled P gate, which are the CTQWs in this work.}
    \label{fig:order-merge}
\end{figure}

%\subsection{Numerical Results}
%\label{sec:performance}
%In this section we present the comparison between the results of numerical simulations obtained for different walk orders of our single-edge state preparation method, the Merging States (MS) method introduced in \cite{Gleinig_2021AnEffAlgoForSparseQSP} and Qiskit's built-in state preparation method based on uniformly controlled rotations \cite{Mottonen_2005TransformOfQuantStsUsingUnifContRot}. The error bars in the figures denote the 95\% confidence interval for the estimate of the mean.

% \subsection{Advantages of this Framework}
Next, we discuss advantages of this framework. The CTQW approach to state preparation allows us to examine the problem from a graph-based point of view. This is a powerful perspective, as it enables us to leverage graph topology, such as connectivity patterns captured by algorithms like the minimum spanning tree, and enforce specific tree structures which can be helpful for developing new state preparation heuristics. For example, we can enforce a linear topology (see Fig.~\ref{fig:all_walks_methods}) or a star shaped topology.
% \subsection{Merging States Method as a CTQW}
Moreover, the state of the art ancilla free sparse state preparation method, referred to as Merging States (MS) method in Ref.~\cite{Gleinig_2021AnEffAlgoForSparseQSP} fits in the CTQW framework we present in this paper. From examination, we determined that MS typically corresponds to star shaped graphs.

Next, we show that the merging states method forms a tree graph representing a CTQW over the basis states of the target state with non-zero amplitude. Since MS eventually results in a single merged state, it must form a tree graph where nodes that are associated with the states being merged all feed into a single node. Thus, the MS is analogous to our splitting method described above but run in reverse. When performing the merging or splitting, the basis states must be brought within a Hamming distance of 1. After merging or splitting of two basis states, the protocol can proceed in the rotated frame or be restored to the original basis. In MS, the protocol proceeds in the rotated frame. An example of MS represented in the CTQW framework is provided in Fig.~\ref{fig:order-merge}. Second, the merging states method relies on a controlled-G gate for merging. This simply corresponds to a particular sequence of CTQWs, given by the pattern of single self-loop, single-edge, and single self-loop.

The follow section  demonstrates that, while the asymptotic scaling is equivalent, the level of sparsity in the state to be prepared can play a large role in the best heuristic to choose. Indeed, there exists the possibility of mixing and matching heuristics during the ordering process to get the best performance across many situations. We leave such an investigation to future work. 

% An additional possibility this framework raises is the expansion of merging beyond two states at a time. 

% The single-edge walks used in this work have been shown to merge/split two basis states, but similar to reducing the controls on controlled single qubit rotations, the number of edges in the walk can be expanded to produce merging/splitting between much larger sets of basis states, a fact we take advantage of in this work to reduce the number of gates be ensuring that the reduced controls result in merging operations merge states with zero amplitude. Since mergig states with zero amplitude is effectively the identity, the result is a saving in the number of gates need to implement the circuit version of the CTQW, but clearly the same principle executed carefully could be useful for state preparation more generally. By reinterpreting state preparation as a sequence of CTQWs, our framework reconstructs known high-quality state preparation methods, produces new ones superior in certain regimes, and contains a multitude of possibilities for further optimization going forward.

% \subsection{Performance}\label{sec:performance}
In this section we present the comparison between the
results of numerical simulations obtained for different
walk orders of our single-edge state preparation method,
the Merging States (MS) method introduced in Ref.~\cite{Gleinig_2021AnEffAlgoForSparseQSP}, and
Qiskit’s built-in state preparation method based on uniformly
controlled rotations Ref.~\cite{Mottonen_2005TransformOfQuantStsUsingUnifContRot, Shende_2006SynthOfQuantLogCirc}. The error bars in the figures
denote the 95\% confidence interval for the estimate
of the mean.

% \subsubsection{Comparison of Walk Orders}

First, we compare the performance of different walk orders presented in Sec.~\ref{sec:methods}, which is shown in Fig.~\ref{fig:all_walks_methods}. The greedy methods use the initial ordering specified in parenthesis. For the ``Greedy Combined" results, we consider both Greedy and the method in parenthesis and choose the one with the lower CX count. 

As seen from the figure, all considered walking orders have similar performance, with Random, Sorted, SHP and MST requiring more controlled gates than the others. MHS Linear and Nonlinear outperform Greedy(Sorted) on average except at 5 qubits where only MHS Linear beats greedy. As expected, Greedy(MHS) Combined performs the best on average.

\begin{figure}[h!]
    \centering
    \includegraphics[width=\linewidth]{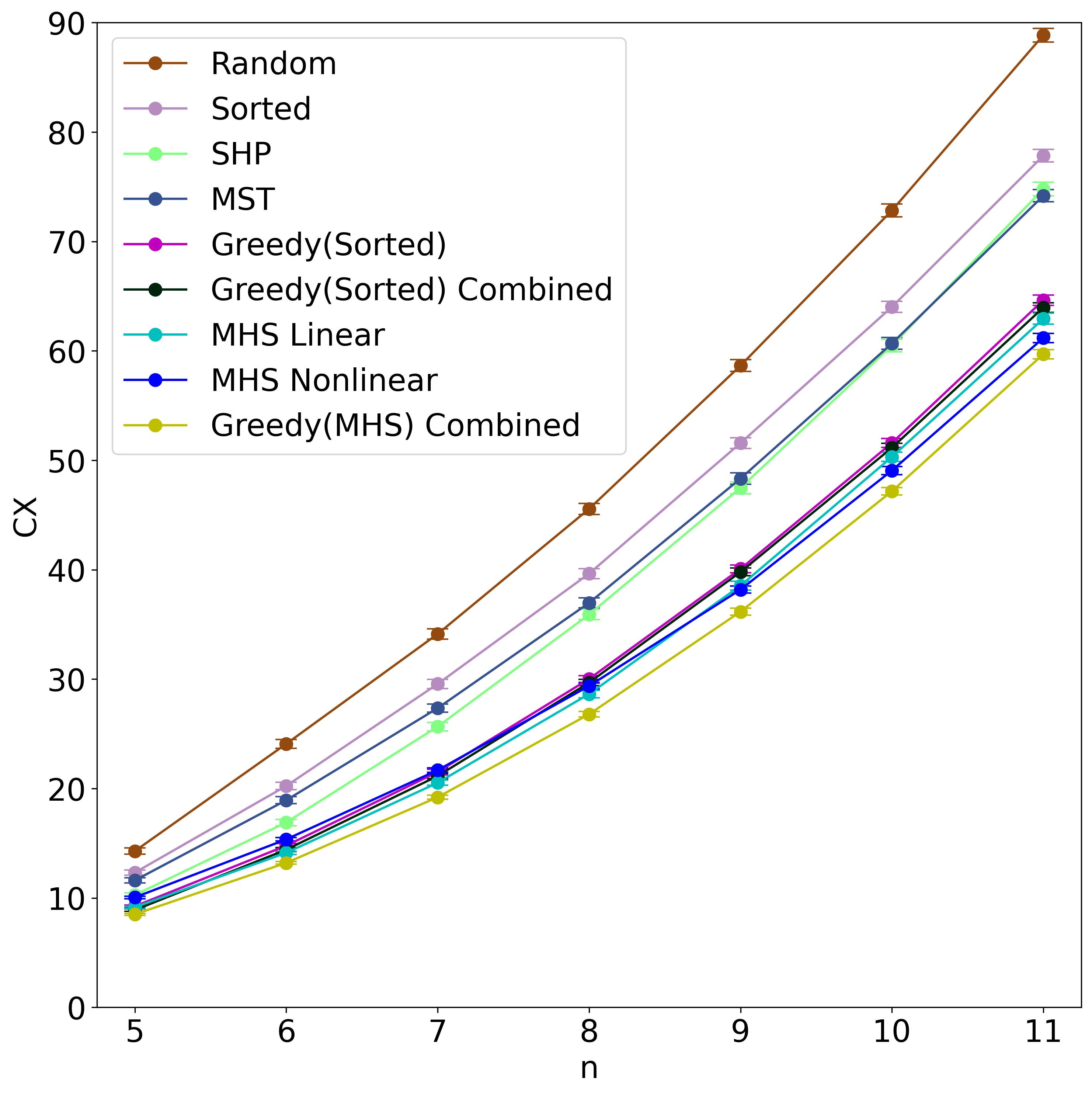}
    \caption{Comparison of walk orders with $m=n$. Greedy(MHS) (Greedy(Sorted)) uses MHS Linear (Sorted) as the initial ordering. Each data point is averaged over 1000 random states.}
    \label{fig:all_walks_methods}
\end{figure}

\begin{figure*}
    \centering
    \subfloat[$m=n$ \label{fig:comparisontoSOA_m=n}]
    {\includegraphics[width=0.5\linewidth]{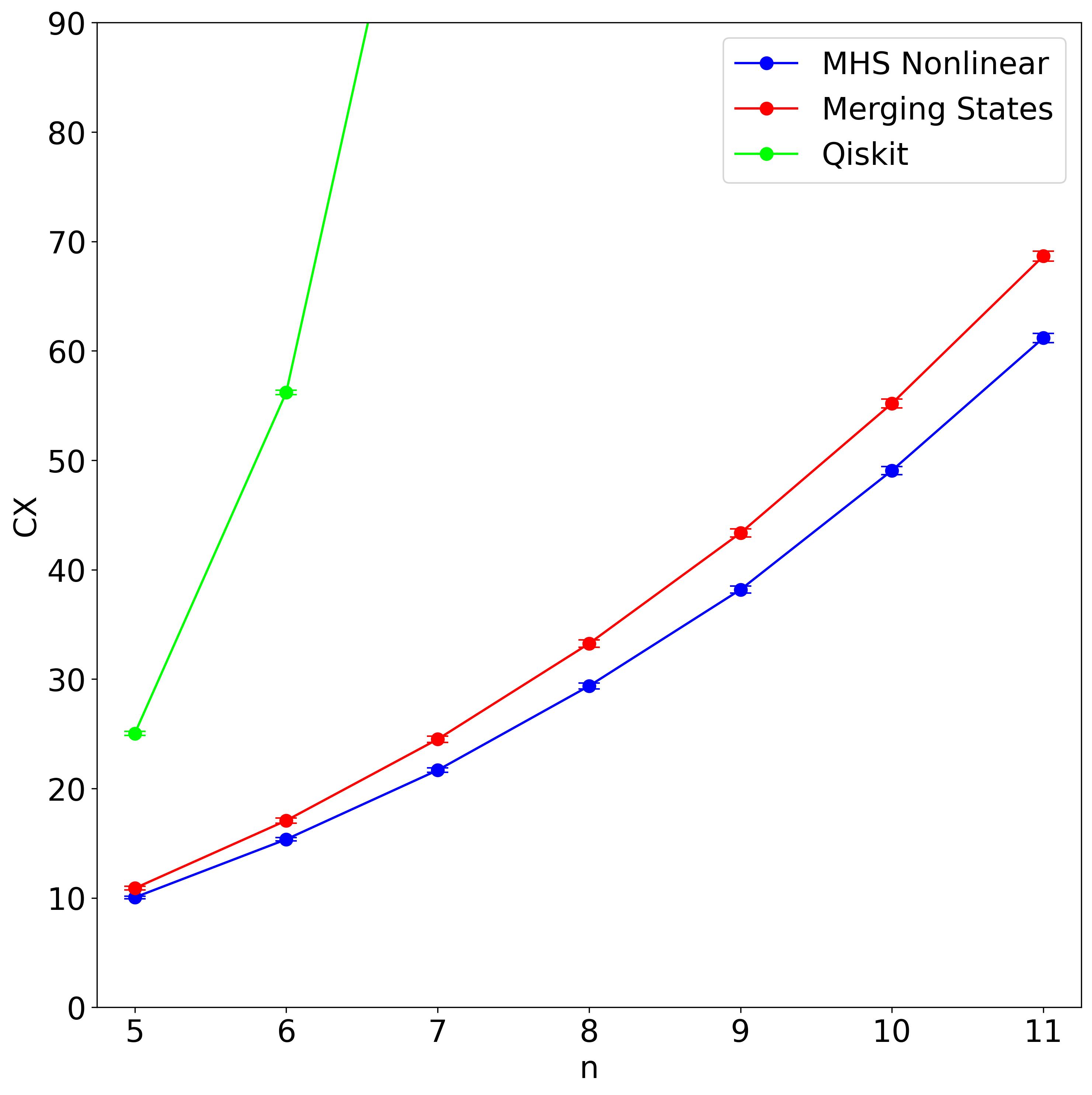}}\hfil
    \subfloat[$m=n^2$. \label{fig:comparisontoSOA_m=n^2}]{\includegraphics[width=0.5\linewidth]{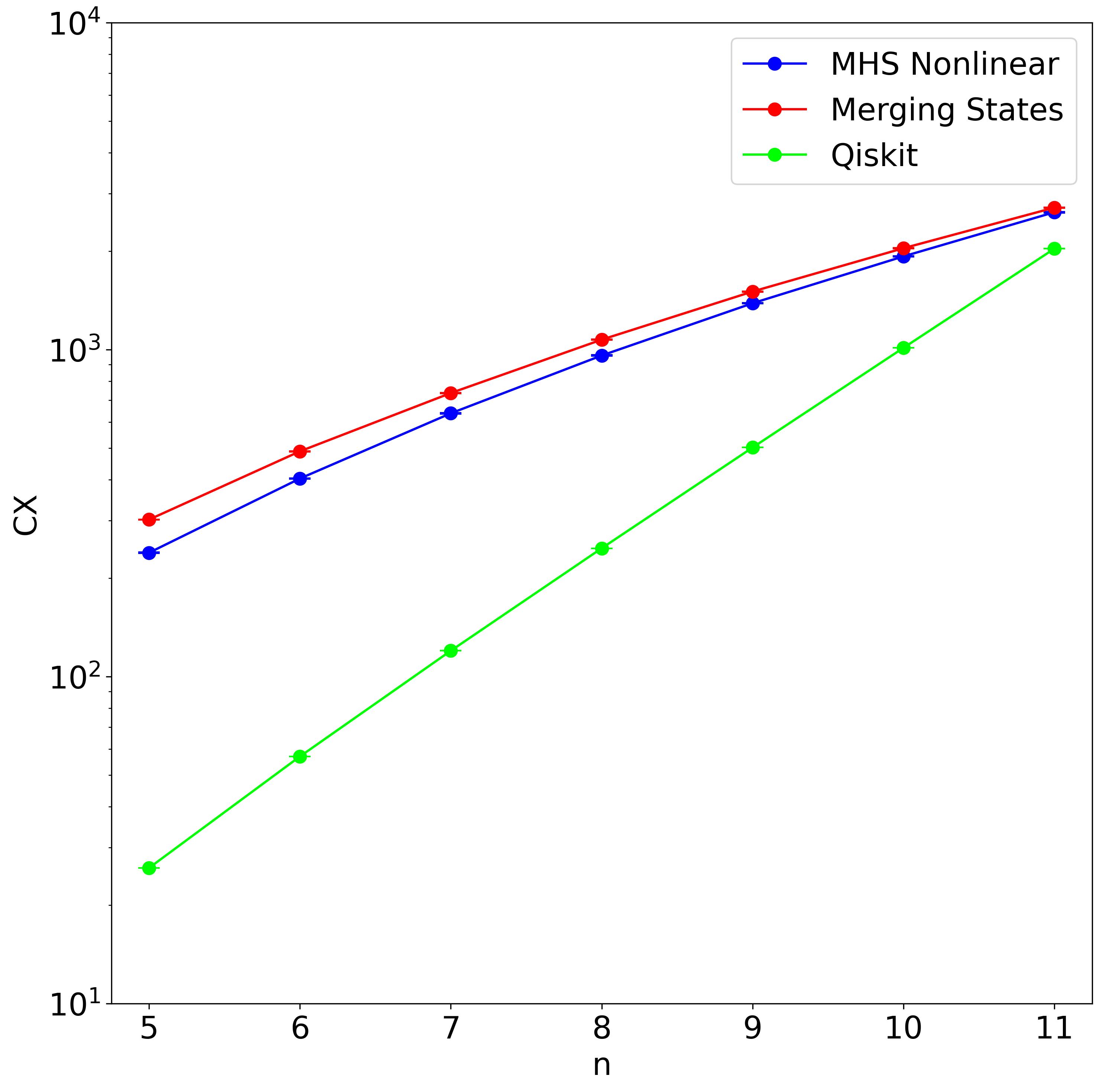}}
    \caption{CX count comparison between single-edge method with the best walk order (MHS Nonlinear) presented in this paper, state-of-the-art sparse state preparation method (Merging States) and Qiskit's built-in state preparation method for (a) $m = n$ and (b) $m = n^2$. Each data point is averaged over 1000 randomly generated states. 
    }
\end{figure*}

% \subsubsection{Comparison to Merging States and Qiskit}
% In contrast to our single-edge state preparation method, which starts from one state and gradually splits existing states to form new states, the MS method of Ref.~\cite{Gleinig_2021AnEffAlgoForSparseQSP} instead focuses on the inverse process of merging basis states.
% The MS method can be interpreted as a series of quantum walks with specified targets for the multi-control gates and optimized with the CX backward-propagation method. Note that in MS, the authors of \cite{Gleinig_2021AnEffAlgoForSparseQSP} use a single multi-controlled gate, denoted CG, to implement the merging (splitting) and the phase correction in one step. However, CG can also be written as a product of CP and CRx gates. 

In Figs.~\ref{fig:comparisontoSOA_m=n} and \ref{fig:comparisontoSOA_m=n^2}, we compare the performance of the best walk order heuristic we found (MHS Nonlinear) to that of the Merging States method and Qiskit's built-in  method (which is based on Ref.~\cite{Shende_2006SynthOfQuantLogCirc}). Just as in the previous section, an average CX gate count for all methods is calculated for a set of 1000 randomly generated sparse states ($m=n$ and $m=n^2$ non-zero amplitude basis states) for every value of $n$ from 5 up to 11.

As can be seen from Figs.~\ref{fig:comparisontoSOA_m=n} and \ref{fig:comparisontoSOA_m=n^2}, Qiskit's performance scales exponentially and is very similar regardless of the value of $m$. In fact, for $m = n^2$ and $m = 2^{n - 1}$ Qiskit produces circuits with exactly the same number of CX gates regardless of the state being prepared.

In contrast to Qiskit, our method takes advantage of the sparsity of the target state and is expected to outperform Qiskit's built-in method for any asymptotically sparse state (i.e. $m = \mathcal{O}(poly(n))$) for sufficiently large values of $n$. Our method also outperforms MS for the considered cases, but our walk order heuristic seems to excel only in the linear case. From numerical examination, it appears that for $m=n$ the gap increases, but for the $m=n^2$ case a crossover happens at $n=13$. All the discussed CTQWs (including MS) have the same asymptotic scaling.

The exact code and data for these numerical experiments can be found in the linked repository (see Data Availability Sec.~\ref{sec:codeAndDataAvail}).

\section{Discussion}
\label{sec:discussion}

In this work, we developed an algorithm that can be used to convert CTQWs on dynamic graphs that consist only of self-loops or single edges to the quantum gate model. This algorithm serves as the basis for a deterministic state preparation framework that has complexity of $\mathcal{O}(nm)$ where $n$ is the number of qubits and $m$ is the number of basis states in the target state. Our framework encompasses MS \cite{Gleinig_2021AnEffAlgoForSparseQSP}.  
Furthermore, we introduce multiple methods that can reduce the CX count of the state preparation algorithm. We test various walk orders against the MS method and uniformly controlled rotation method \cite{Shende_2006SynthOfQuantLogCirc}. We find that our MHS Nonlinear heuristic excels and requires fewer CX gates than the MS method when the target state has a linear number of non-zero amplitudes. 

% We also show that MS is encompassed in our framework.

% In Ref.~\cite{Mozafari_2022EffDetPrepQuantStUsingDecisionDiag}, the authors present another sparse state preparation method that requires only 1 ancilla. We did not directly compare to their method, but their protocol is outperformed by the MS method for low-degree polynomials $m(n)$ \cite{Mozafari_2022EffDetPrepQuantStUsingDecisionDiag}. 
% Since our method outperforms the MS method in these regimes, our method should outperform theirs as well.

The single-edge framework we present here posts many interesting possible avenues for future investigation. As mentioned in Section~\ref{sec:methods}, the gate count for preparing an arbitrary state can be reduced depending on the order in which the state transfer occurs. Without control reduction, the optimal sequence of walks is given by the minimum spanning tree approach. Determining the optimal sequence of basis state transfers in the presence of control reduction is more complicated and could be an interesting direction for future research. Additionally, if a target state is comprised of basis states with symmetry, it may be possible to cleverly reduce the number of gates required to create the target state by exploiting symmetry.

\section{Methods}
\label{sec:methods}

% In this section, we provide background information and examples of CTQWs on dynamic graphs and state-of-the-art state preparation circuits. We then 

% \subsubsection{Control Reduction}
% \label{sec:control_reduction}
One of the most powerful optimization methods we utilize is control reduction.
The number of CX gates necessary to implement the multi-controlled Rz and Rx gates is proportional to the number of controls on those gates. As mentioned earlier, we might need up to $n - 1$ controls in the worst case, but depending on the walk that we want to implement and the basis states that we have already visited, we might need less than that. For example, when we perform the first walk from the root, we never need any controls, since there are no other basis states that would be affected by the Rx gate.

More generally,
\begin{prop}
\label{prop:contrl_red}
    Let $n$ denote the number of qubits and $S = \{\ket{z_1}, \ket{z_2}, ..., \ket{z_k}\}$ denote the set of visited nodes. Suppose we wish to perform a single-edge walk from $\ket{z_j}\in S$ to $\ket{z_{\ell}} \notin S$ and $\ket{z_j}$ and $\ket{z_{\ell}}$ have a Hamming distance of 1. Let $b$ denote the qubit where $\ket{z_j}$ is different from $\ket{z_{\ell}}$, i.e. $z_j[b] \neq z_{\ell}[b].$  Let $D=\{\{k \, | \, z_i[k]\neq z_j[k], k\neq b \} \, |\, z_i\in S\setminus z_j\}$ denote the set of differing bits of the visited nodes with $z_j$ excluding the qubit $b$. For the gate-based representation of a single-edge walk from basis state $\ket{z_j}$ to $\ket{z_{\ell}}$, it is sufficient to control the CRz or CRx gate on any hitting set of $D$, where the values of controls are equal to the corresponding bits of $z_j$. A hitting set of $D$ is a collection of elements $h$ such that $h \cap d_i \neq \emptyset$ for all $d_i \in D$.
\end{prop}

\begin{proof}
    Adding a control on an arbitrary qubit $c$ to any gate makes the gate act only on the states conforming to the value of that control, i.e. $\ket{z_k}$ such that $z_k[c] = 0$ or 1, depending on the state of the control. By the definition of $d_i \in D$, controlling the CRz or CRx gate on any index $e \in d_i$ with the value of control equal to $z_j[e]$ will ensure that the gate does not act on $z_i$. Thus, controlling on any hitting set of $D$ will ensure that that the gate does not act on all $z_i \in S$, except $z_j$.
\end{proof}

When $z_j$ and $z_{\ell}$ have a Hamming distance greater than 1, we first have to update $S\rightarrow \tilde S$ with the conjugating CX gates, as described in Sec.~\ref{sec:results}. Then we apply the previously described control reduction method to $\tilde S$. Note that in this case, any of the differing qubits can be used as a target for the CRz or CRx gates, with the appropriate CX conjugation. However, the target qubit cannot be chosen as a control, therefore the choice of target affects the efficiency of control reduction. In the following sections, when a target qubit is not explicitly dictated by a particular walk order, we consider all options for the target and choose the one that results in the minimum number of controls.

In order to minimize the number of controls, one needs to minimize the size of the hitting set, i.e. solve the minimum hitting set (MHS) problem, which is known to be NP-complete. As such, no polynomial algorithm is known to solve it exactly, but a number of heuristics exist that can provide good suboptimal solutions quickly \cite{vinterbo2000minimal, gainer2017minimal}.
In practice, solving this problem for each walk drastically reduces the number of control qubits on many gates in the sequence (see Fig.~\ref{fig:control_reduction}). The MHS problems that we consider in the following sections are solved approximately.

\textbf{Example: } 
Consider the sequence of walks [[001, 111], [111, 110]]. Suppose we have implemented the first walk and are now implementing the second walk. The second walk transfers amplitude from 111 to 110. The current state is a superpostion of $\ket{001}$ and $\ket{111}$. Thus, we want the walk to affect only the amplitude of basis state 111 and leave 001 alone. Therefore, we need the controls for the Rx gate to be satisfied only by 111. Controlling the CRx gate on qubit 0 accomplishes this.

\begin{figure}
    \centering
    \includegraphics[width=0.42\textwidth]{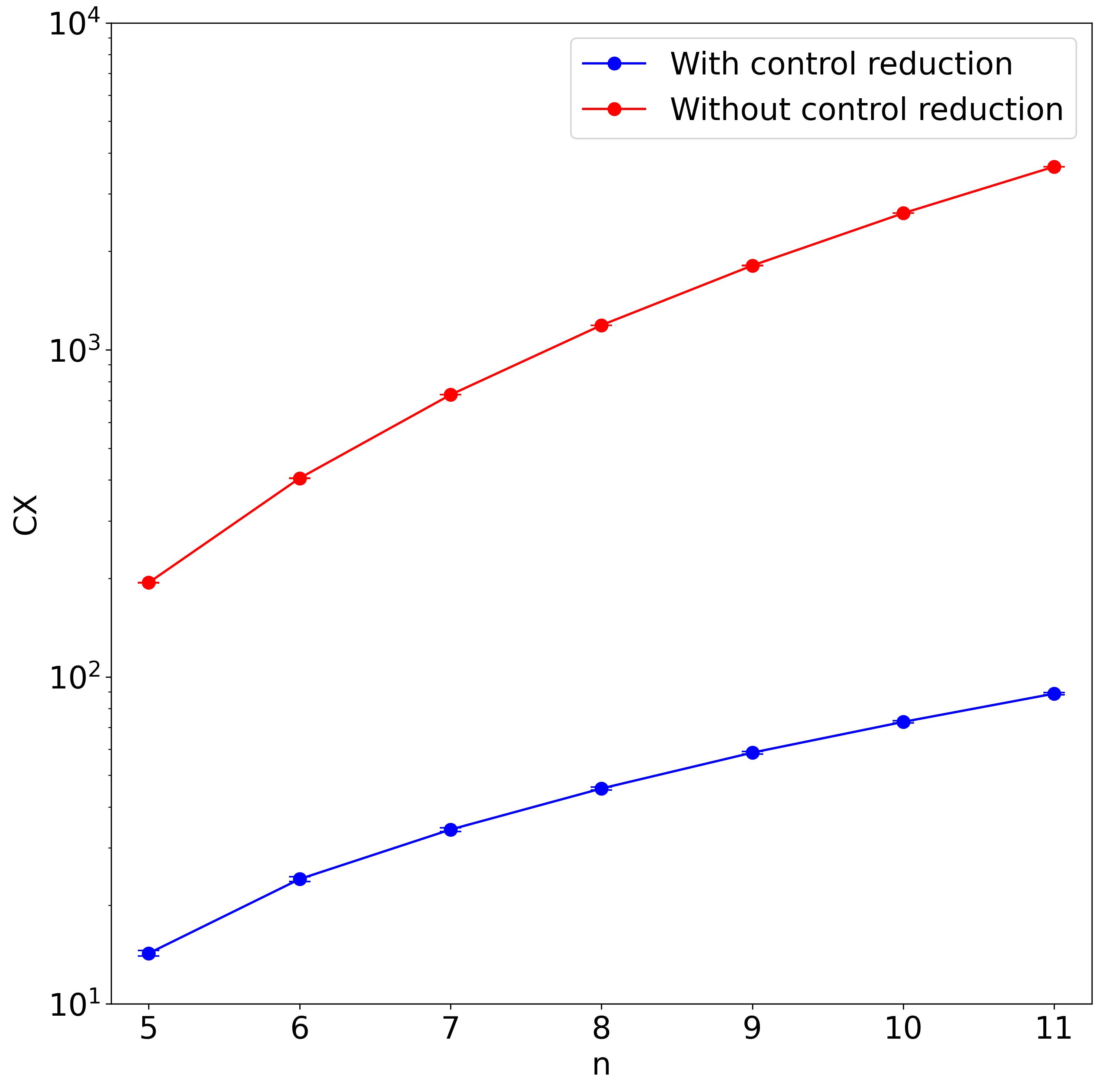}
    \caption{The average number of CX gates necessary to implement a random state with $n$ qubits with and without control reduction. Each state consist of $m = n$ non-zero amplitudes and each point is the average over 1000 random states.}
    \label{fig:control_reduction}
\end{figure}

% \subsubsection{CX Forward- and Backward-Propagation}
Next, as described in Sec.~\ref{sec:results}, single-edge quantum walks between states with Hamming distance greater than one require conjugation with CX gates on both sides of the CRx gate. The CX gates on the left side bring the desired basis states to a Hamming distance of 1, where they interact via the CRx gate, and the ones on the right move the basis states back to their original positions. Rather than immediately moving the states back after the application of the CRx gate, one can instead keep working with the transformed basis and move the right branch of the conjugating CX gates to the end of the circuit, which allows many of those CX gates from different single-edge walks to cancel. This is called forward-propagation of CX gates. The backward-propagation of CX gates is the same method applied to the reverse of the circuit. The advantage of the backward-propagation is that the CX gates are moved to the beginning of the circuit. Since the initial state is the ground state, these CX gates act trivially and can be removed.

\textbf{Example: }
Consider the sequence of walks [[001, 111], [011, 110]]. The first walk [001, 111] has a Hamming distance of 2. For target qubit 1, we require CX$_{10}$ to bring the states to a Hamming distance of 1. Applying this to the basis states we get [001, 011], [111, 010]. After applying the multicontrol gate(s) for the first walk, we do not rotate the basis states back to the original frame, but instead keep track of the CX rotation which will be applied at the end of the circuit. This is an example of CX forward-propagation.

% \subsubsection{Leaf Node Considerations}
% \label{sec:leaf_gate_prxp}
Next, we discuss leaf node considerations.
% \textcolor{red}{When a walk is performed between two non-leaf nodes of the tree (see Fig.~\ref{fig:sparse-subset} UPDATE THIS ONCE NEW FIGURE 4 IS INSERTED), we use a CRz gate to adjust the phase on the walk origin node, followed by a CRx gate to transfer the amplitude to the destination node, both of which have the same control set, since the CRz operation does not change the number of non-zero amplitude basis states}. 
The phase on the destination node is adjusted once we consider the first walk out of that node. However, when the destination node is a terminal (leaf) node, there are no other walks that start from it and we still need to adjust its phase. 

One way to solve this issue is to use a separate CRz gate, but it adds another multi-controlled gate to the circuit and the number of controls for it is likely to be greater compared to the walk into the leaf node due to the fact that we will need to distinguish the leaf from a larger number of non-zero amplitude states and control reduction will be less efficient.

An alternative approach is to combine this last CRz gate with the CRz and CRx gates from the walk leading into the leaf, as one gate that simultaneously transfers the probability amplitude and adjusts the phase on both ends of the walk. This reduces the number of multi-controlled operations from 3 to 1, but the downside of this approach is that the resulting gate will not belong to SU(2) group. As such, the number of CX gates in its decomposition will be proportional to $\mathcal{O}(n^2)$, instead of $\mathcal{O}(n)$ necessary for the regular SU(2) gates, such as CRz or CRx. \cite{Vale_2024CircDecompOfMulticontSU(2)} Asymptotically, it makes this approach less efficient for sufficiently large values of $n$, but for our dataset it is still favorable.

% \subsubsection{Tree Construction and Walk Order}
% \label{sec:walk_order}
Any tree constructed in the first step is, in principle, sufficient for state preparation, but not all trees result in the same CX gate count. Moreover, different walk orders on the same tree can can result in different number of CX gates as well. 
In this section we consider some of the options for tree construction and walk order.

% \vspace{5pt}
% \textbf{Best Orders.}
We first describe heuristics that typically perform best on average.

\vspace{5pt}
\textit{MHS Nonlinear.}
As discussed previously, control reduction is a powerful technique that allows to significantly simplify the state preparation circuit. The efficiency of it depends on how easily a particular state can be separated from the rest of the currently existing states with non-zero amplitudes, which in turn depends on the local density of states around the state that we want to separate. The basic idea of the following walk order heuristic is that to maximize the efficiency of control reduction, we need to walk through hard-to-separate states first, before the density of states around them increases, which allows us to separate them using smaller number of controls. Since the target qubit has a significant effect on the CX count, it is determined during the walk order construction as well.

More specifically, we construct a sequence of walks in reverse time order, from last to first. Given a list of basis states, we choose the state that requires the smallest number of controls to differentiate it from the rest of the basis states by solving the Minimum Hitting Set (MHS) problem. When two states are equally easy to distinguish, we choose the state with the largest number of bits that are different from the rest. Let us call this state $z_1$. The target qubit $t$ is the MHS bit with the smallest frequency. The set differentiated by $t$ is then ordered by the easiest element to distinguish via the MHS. When two elements are equally easy to distinguish, we order them by Hamming distance from $z_1$. The easiest element is set as $z_2$. The ($z_1$, $z_2$, $t$) that requires the fewest controls is added to the path. The number of controls required is upper bounded by the size of the MHS for $z_1$ plus the size of the MHS for $z_2$. 

All the elements are then updated by the back-propagation of CX gates. $z_2$ is removed from the list and the process is repeated until every node is visited by the path. The circuit is constructed from the path of CTQWs in terms of the original basis states and target qubits and with CX backward-propagation (see Alg. \ref{alg:mhs_nonlinear} for pseudocode).

\vspace{5pt}
\textit{MHS Linear.}
This protocol follows the same procedure as MHS Nonlinear except that each step we assign $z_1$ of the previous pair as the $z_2$ for the next pair. This results in a linear graph. In contrast, in MHS Nonlinear $z_2$ of the next pair is chosen using the MHS methods.

\vspace{5pt}
\textit{Greedy Insertion.}
The linear greedy algorithm takes an ordered list of basis states and sequentially finds the best position to insert the next basis state. The classical complexity of this algorithm is $\mathcal{O}(nm^4+C)$, where $C$ is the classical complexity of the method that generates the initial ordering. Typically, the greedy insertion performs better when given an initial ordering that already performs well (see Alg. \ref{alg:greedy} for pseudocode).

% \vspace{5pt}
% \textbf{Other Orders.}
The following are other orders we considered that on average typically do not outperform state-of-the-art methods.

\vspace{5pt}
\textit{Minimum Spanning Tree.}
The most straightforward option to minimize the Hamming distance between the connected basis states is to build a Minimum Spanning Tree (MST) over the complete graph of target basis states, where each edge is weighted by the Hamming distance between its endpoints.

The cost of calculating all pairwise Hamming distances is $\mathcal{O}(nm^2)$, while the MST itself can be built in $\mathcal{O}(m^2)$ with Prim's algorithm. Thus, the overall classical complexity for this method is $\mathcal{O}(nm^2)$.

In the case when the target state is fully dense, i.e. $m = 2^n$, this method can be simplified, since the tree can be automatically built without calculating all pairwise distances by branching off sequentially in each dimension of the hypercube (see Fig.~\ref{fig:example_walks}b). 
The same method can also be used for generally dense states, i.e. when $m = \mathcal{O}(2^n)$. In this case, the resulting tree will go through some zero-amplitude states, which is not optimal for the circuit, but allows to save classical computational resources.

\vspace{5pt}
\textit{Shortest Hamiltonian Path.}
The MST approach minimizes the total Hamming distance for the walks, which corresponds to the optimal walk order if control reduction is not taken into account. However, when control reduction is applied, other walk orders may be more efficient since they may allow for additional control reduction.

Therefore, instead of MST, one could find the Shortest Hamiltonian Path (SHP) in the same graph of the target basis states as used to build the MST. This option is more computationally expensive, since the cost of finding SHP exactly is $\mathcal{O}(m!)$. However, approximate SHP heuristics can find suboptimal solutions much faster \cite{gurevich1987expected}.

Similar to MST, when $m = 2^n$ the procedure can be simplified by building a Hamiltonian path throughout the whole hypercube of basis states without calculating all pairwise distances. This path is the known as the Gray code.
Examples of the quantum walks produced by SHP for the cases of sparse and dense states are shown in Fig.~\ref{fig:example_walks}.

\vspace{5pt}
\textit{Sorted order.}
Another method to build a simple linear path without solving SHP is to connect the basis states sequentially in increasing order.
This method is less optimal than SHP, but it is computationally inexpensive since the target basis states can be sorted in $\mathcal{O}(nm\log(m))$, and can provide us with a reasonably good walking order. 
For example, for a fully dense state, the average Hamming distance given by the sorted walk order is 2, which is much better than $n / 2$ achieved by a random walk order. 

A comparison of the different walking orders presented in this section for the case of $m = n$ 
%with control reduction 
is shown in Fig.~\ref{fig:all_walks_methods}. We empirically validate the quantum walk order methods and find that on the prepared states, the MHS methods and greedy methods require the fewest CX gates. Greedy(MHS) Combined is the greedy method on the initial ordering of MHS Linear and the better circuit between the two methods is used.

\begin{figure}
    \centering
    
    \subfloat[MST for a sparse state.]{\includegraphics[width=0.2\textwidth]{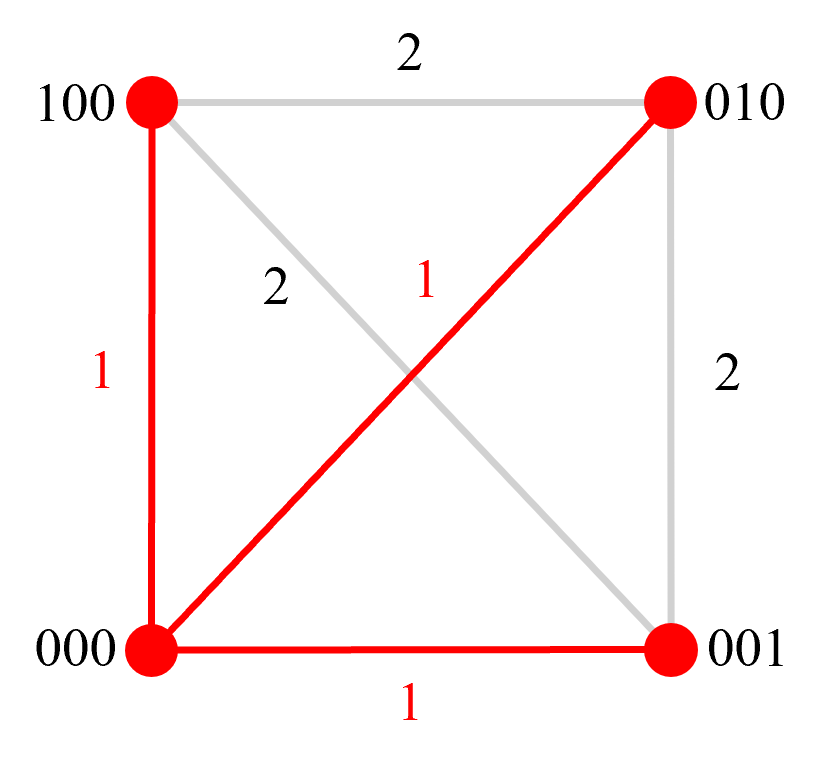}}
    % \hspace{10pt}
    \subfloat[\label{fig:mst_ex_dense} MST for a dense state.]{\includegraphics[width=0.25\textwidth]{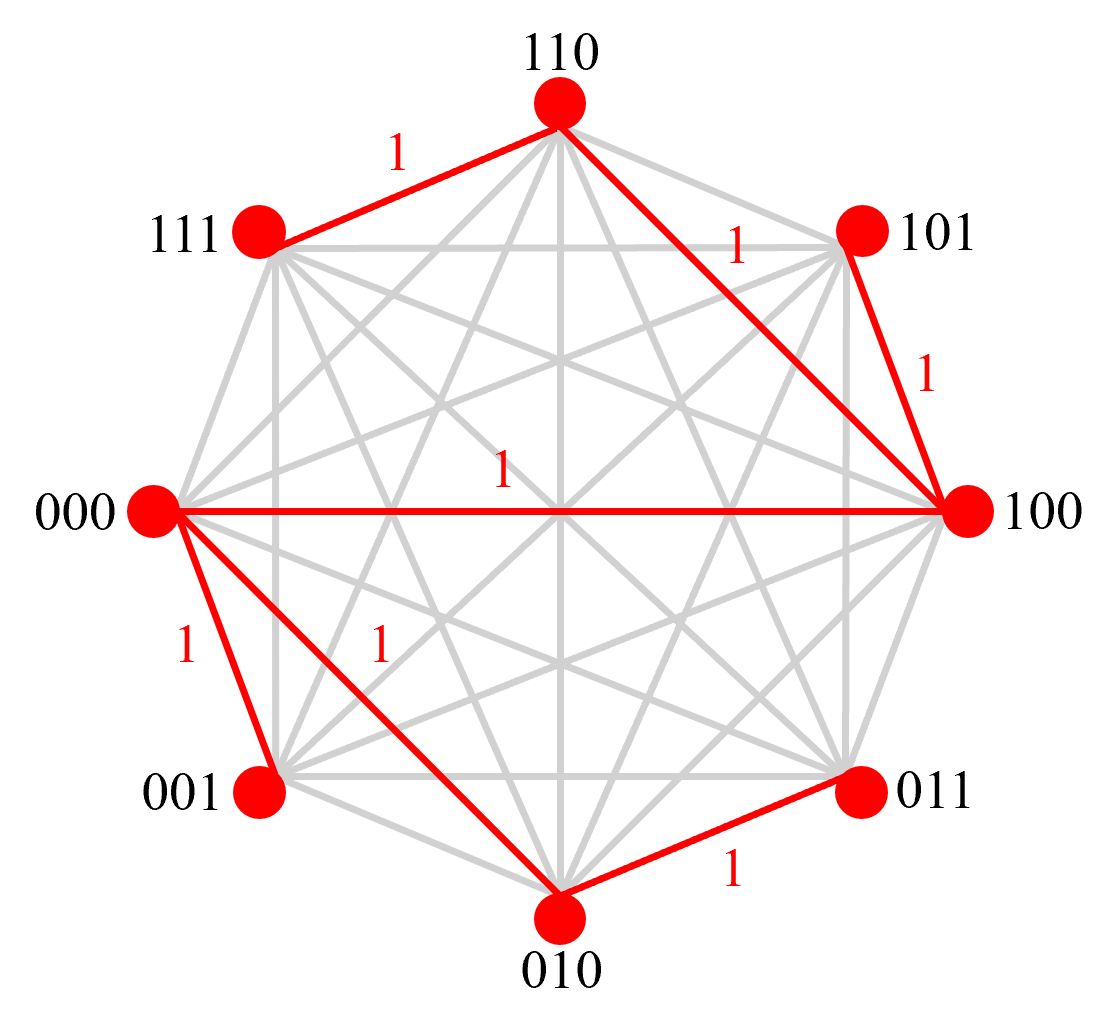}}
    
    \subfloat[SHP for a sparse state.]{\includegraphics[width=0.2\textwidth]{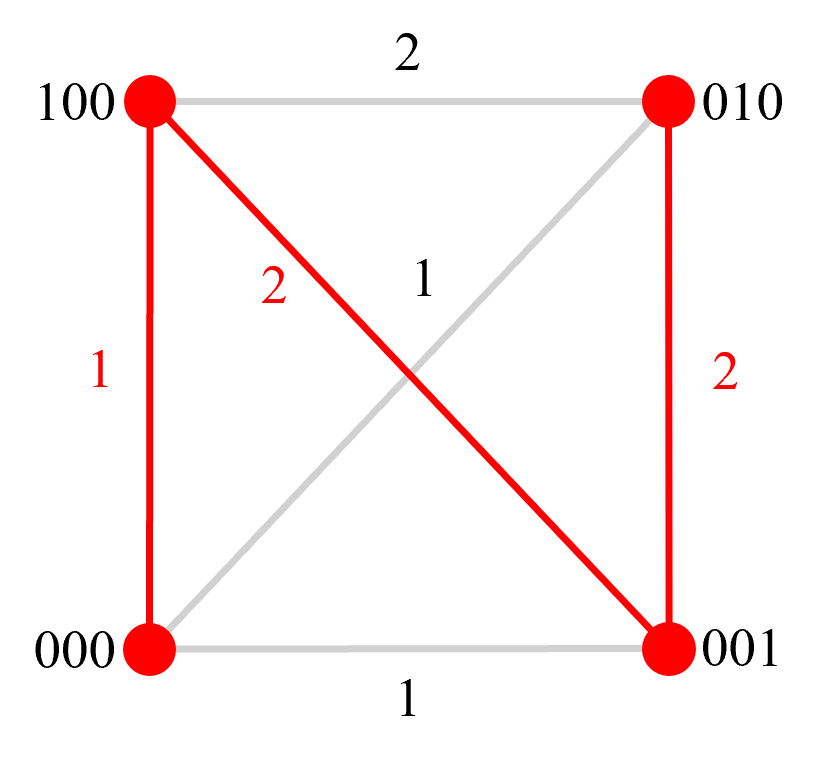}}
    % \hspace{10pt}
    \subfloat[\label{fig:shp_ex_dense}SHP for a dense state.]{\includegraphics[width=0.25\textwidth]{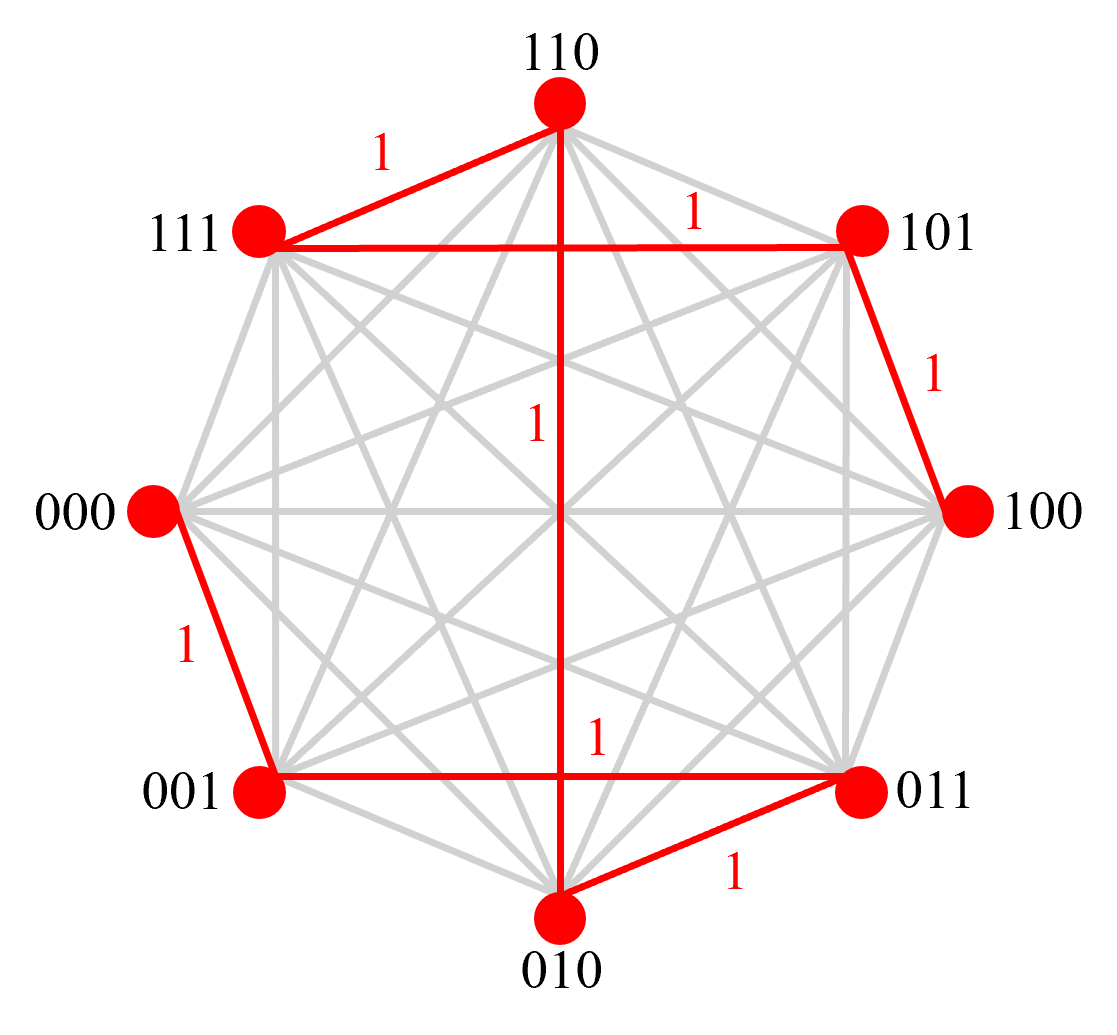}}
    
    \caption{Example walks generated by different methods for the case of 3 qubits. The weight of an edge is the Hamming distance of the two nodes. A potential walk path chosen by a given method is shown in red. In Figs.~\ref{fig:shp_ex_dense} and ~\ref{fig:mst_ex_dense}, the costs of the edges are suppressed except for the path.}
    \label{fig:example_walks}
\end{figure}

% \subsubsection{Additional Optimizations}
% \label{sec:AdditionalOptimizations}

Lastly, a single edge between states $\ket{j}$ and $\ket{k}$ corresponds to a CRx gate, which introduces an imaginary phase on $\ket{k}$. However, for the purpose of preparing a state with some real-valued amplitudes, a better approach is to use a CRy gate instead, which has the same action as CRx, except it does not introduce a complex phase. CRy keeps the amplitudes real and makes it unnecessary to use CRz or CP gates when the amplitude is transferred between the basis states with real coefficients. As such, for these states, one might want to use a special walk order that goes through all real-valued basis states before moving to the complex-valued ones. For the states where all amplitudes are real, this approach can reduce the total number of CX gates in the circuit by a factor of approximately 2. We plan to implement this in our code in the future.

% \appendix

\begin{algorithm}\label{alg:conversion}
\caption{Conversion}
\begin{algorithmic}[1]
    \STATE Convert from CTQWs to a circuit.
    \STATE \textbf{Input:}\\
    $walks$: The set of walks.\\
    $state$: The target state.
    \STATE \textbf{Output:}\\
    The circuit that implements the quantum walks.
    \STATE
    % \Function{Cnot}{$op1$}
    %     \State \Return lookup CNOT$op1$CNOT
    % \EndFunction
    \STATE $walks \gets$ Reverse of the walks.
    \STATE $circ \gets$ Initial empty quantum circuit.
    \FOR {$z_1$, $z_2$, $target$, $pangle$, $xangle$ in walks}
    \IF {$target$ is not specified}
    \STATE $target\gets$ \Comment{The qubit index that results in the smallest number of controls.}
    \ENDIF
    \STATE \{Apply the CX gates required to bring $z_1$ and $z_2$ to a Hamming distance of 1.\}
    
    \STATE Determine $mhs$ for $z_1, z_2, target$.
    \IF{$z_2$ is a leaf node}
    \STATE \{Apply controlled U Gate to $circ$ (corresponds to the phase, Rx, phase gate pattern described in the leaf node discussion in Sec. \ref{sec:methods}.)\}
    % \ENDIF
    \ELSE
    \STATE Apply controlled SU(2) Gate to $circ$.
    \ENDIF
    \STATE Update the walks with the CX gates.
    % \Function{Convert}{$op1$, $op2$}
    %     \IF{$op1$ is X}
    %         \STATE \Return lookup $X(op2)X$
    % \EndFunction
    \ENDFOR\\
    \Return inverse of $circ$.
\end{algorithmic}
\end{algorithm}

\begin{algorithm}
\label{alg:mhs_nonlinear}
\caption{MHS Nonlinear}
\begin{algorithmic}[1]
    \STATE \textbf{Input:}\\
    $basis$: A set of basis states to walk through.
    \STATE \textbf{Output:}\\
    A sequence of single-edge quantum walks [[$z_1$, $z_2$, $target$]$_i$] through the provided basis states that minimizes the number of controls needed to implement the corresponding multi-controlled operations.\\
    % \STATE $basis \gets$ basis states
    \STATE $path \gets$ empty.
    \STATE $mutable\_basis \gets$ copy $basis$.
    \WHILE{size of $basis>1$}

        \STATE $all\_mhs\_z_1 \gets$  \Comment{MHS for every element in $mutable\_basis$.}
        \STATE $all\_diffs\gets$ \Comment{for every element $z_i$ in $mutable\_basis$ get the subset $k_i=mutable\_basis/z_i$ and construct a list of list of the bits that differ with $z_i$.}
        \STATE $all\_targets \gets$ \Comment{For every element $mhs$ in $all\_mhs\_z_1 $ get the smallest frequency index in the corresponding $all\_diffs[mhs]$}.
        \STATE $all\_basis\_z_2\gets$ \Comment{A list of list. For every element $z_i$ in $mutable\_basis$ get the subset $k_i\subset mutable\_basis/z_i$ such that $k_i$ is \textit{only} covered by the hitting set $all\_mhs\_z_1[z_i]$ with the corresponding target $all\_targets[z_i]$.}
        \STATE $all\_mhs\_z_2 \gets$ \Comment{A list of list of list. For every element $basis_i$ in $all\_basis\_z_2$ get the MHS for each element in $basis_i$.}
        \STATE \{Append to $path$ $[z_1, z_2, target]$ that corresponds to the smallest size($all\_mhs\_z_1[z_1]$)$+$size($all\_mhs\_z_2[z_1][z_2]$). If pairs have the same size, use the one with the largest $all\_diffs[z_1]$). Note that $z_1$ and $z_2$ are in terms of the original basis.\}
        \STATE Remove $z_2$ from $basis$.
        \STATE Remove $z_2$ from $mutable\_basis$.
        \STATE Update $mutable\_basis$ with CX gates.
        
        % % \If{$Found$}
        % %     \State BREAK
        % % \EndIf
        % % \State $optimalWeight\gets \text{weight of } pauli$
        % % \State $c_2 \gets pauli$
        % \State $op2\gets pauli$
        % \For{$op1$ and $index$ in $circ$}
            
        %     \If{\Call{CanContinue}{$op1$, $op2$}}
        %         \State $op2 \gets$ \Call{Push}{op1, $op2$}
        %     \Else
        %         \State BREAK \Comment{This check attempt didn't work so break out and try the next Pauli.}
        %     \EndIf
        %     \If{$index$ is the last index}
        %         \State $layersFound \gets layersFound+1$
        %         \State $c_1 \gets$ append $op2$
        %         \State $c_2 \gets$ append $pauli$ 
        %     \EndIf
        % \EndFor
    \ENDWHILE\\
    \Return the inverse of $path$.
    % \EndWhile
\end{algorithmic}
\end{algorithm}

\begin{algorithm}\label{alg:greedy}
\caption{Greedy Insertion}
\begin{algorithmic}[1]
    \STATE The walks considered here all form linear graphs. Greedy finds the best position to insert the next basis state into the path sequence.
    \STATE \textbf{Input:}\\
    $basis$: Basis in some order.
    \STATE \textbf{Output:}\\
    The best circuit found.\\

    \STATE $path \gets$ Empty list
    % \STATE $circ \gets$ Empty circuit
    \FOR{$z$ in $basis$}
        \STATE $cxcount\gets$ None.
        \STATE $curr\_path\gets$ None.
        \FOR{$idx$ in (length($path$)+1)}
        \STATE $temp\_path \gets$ copy $path$.
        \STATE Insert $z$ in $temp\_path$ at $idx$.
        \STATE $temp\_cxcount\gets$ CX count for constructed circuit with $temp\_path$.
        \IF{$cxcount$ is None or $temp\_cxcount < cxcount$}
        \STATE $curr\_path\gets temp\_path$.
        \STATE $cxcount\gets temp\_cxcount$.
        \ENDIF
    % \textbf{function} CanContinue($op1$, $op2$)\\
    %     \bindent
    %     \STATE\Return $op1$ is not  or ($op2$ is I or Z)
    %     \eindent\\
    % \textbf{end function}
    \ENDFOR
    \STATE $path\gets curr\_path$.
    \ENDFOR\\
    \Return circuit for $path$.
\end{algorithmic}
\end{algorithm}

\section*{Code and Data Availability}
\label{sec:codeAndDataAvail}
The code and data for this research can be found at \url{https://github.com/GaidaiIgor/quantum_walks}.

\section*{Acknowledgments}
The authors would like to thank Mostafa Atallah for useful conversations regarding this work. 

AG, JL, and ZHS acknowledge DOE-145-SE-14055-CTQW-FY23. CC and TT acknowledge DE-AC02-06CH11357. IG and RH acknowledge DE-SC0024290. 
This research used resources of the Argonne Leadership Computing Facility, a U.S. Department of Energy (DOE) Office of Science user facility at Argonne National Laboratory and is based on research supported by the U.S. DOE Office of Science-Advanced Scientific Computing Research Program, under Contract No. DE-AC02-06CH11357. The funder played no role in study design, data collection, analysis and interpretation of data, or the writing of this manuscript.

\section*{Author contributions}
AG developed the CTQWs to gates conversion, formulated control reduction, wrote code, gathered data, constructed proofs, and calculated complexities. 
RH developed the algorithm that converts CTQWs on dynamic graphs to the circuit model. 
CC wrote code, developed state preparation approaches, and tested circuit optimization techniques.
IG developed the minimum spanning tree approach, CRz implementation of the self-loop walks, wrote code, collected data and prepared some figures.
JL also developed the circuit optimization for reducing control bits. 
TT helped develop state preparation approaches and evaluation methodology.
ZHS had the idea of creating a dictionary between single-edge quantum walks and quantum circuits. 

All authors wrote, read, edited, and approved the final manuscript. 

\section*{Competing interests}
All authors declare no financial or non-financial competing interests. 

\vspace{20pt}
\noindent
\framebox{\parbox{0.97\linewidth}{
The submitted manuscript has been created by UChicago Argonne, LLC, Operator of 
Argonne National Laboratory (``Argonne''). Argonne, a U.S.\ Department of 
Energy Office of Science laboratory, is operated under Contract No.\ 
DE-AC02-06CH11357. 
The U.S.\ Government retains for itself, and others acting on its behalf, a 
paid-up nonexclusive, irrevocable worldwide license in said article to 
reproduce, prepare derivative works, distribute copies to the public, and 
perform publicly and display publicly, by or on behalf of the Government.  The 
Department of Energy will provide public access to these results of federally 
sponsored research in accordance with the DOE Public Access Plan. 
http://energy.gov/downloads/doe-public-access-plan.}}

%\bibliographystyle{unsrt}
%\bibliography{refs}

\end{document}